\documentclass[a4paper, reqno]{amsart}
\usepackage{amsthm}
\usepackage{amsmath}
\usepackage{geometry}
\usepackage{amsfonts}
\usepackage{graphicx}
\usepackage{array}
\usepackage{amssymb}
\usepackage{mathrsfs}

\usepackage{pgf}
\usepackage{tikz}
\usetikzlibrary{arrows}
\usetikzlibrary{cd}
\usetikzlibrary{matrix}
\usetikzlibrary{positioning}
\usepackage{color}
\usepackage{empheq}

\usepackage{amsaddr}

\newtheorem{theorem}{Theorem}

\begin{document}

\markboth{A. Bautista, A. Ibort, J. Lafuente}{Sky conformal invariant}

\title{The sky invariant: A new conformal invariant for Schwarzschild spacetime}

\author{A. BAUTISTA}
\address{\small Depto. de An\'alisis Econ\'omico: Econom\'{\i}a Cuantitativa, Univ. Aut\'onoma de Madrid \\ C/ Francisco Tom\'as y Valiente 5, 28049 Madrid, Spain. }
\email{alfredo.bautista@uam.es} 

\author{A. IBORT}
\address{\small Depto. de Matem\'aticas, Univ. Carlos III de Madrid \\ Avda. de la Universidad 30, 28911 Legan\'es, Madrid, Spain, and \\ ICMAT, Instituto de Ciencias Matem\'{a}ticas (CSIC-UAM-UC3M-UCM)\\ C/ Nicol\'as Cabrera, 13-15, 28049, Madrid, Spain. }
\email{albertoi@math.uc3m.es}

\author{J. LAFUENTE}
\address{\small Depto. de Geometr\'{\i}a y Topolog\'{\i}a, Univ. Complutense de Madrid \\ Avda. Complutense s/n, 28040 Madrid, Spain. }
\email{jlafuente@mat.ucm.es}

\begin{abstract}
A new class of conformal invariants for a given spacetime $M$ is introduced exploiting the conformal geometry of any light ray $\Gamma$.   Each congruence of light rays passing through a given point $p$ defines the sky $S(p)$ of such point.  The new conformal invariants are defined on the bunlde of skies of the spacetime $M$, being called sky invariants accordingly.   The natural conformal covariant derivative defined on a light ray and its associated covariant calculus allows us show the existence of a natural conformal invariant differential of arc that, together with the restriction of the curvature of the conformal covariant derivative, can be used to construct a sky invariant that will be called the sky curvature.   An algorithm, that can be implemented on any symbolic manipulation software system, to compute the sky curvature will be discussed and the main ideas and the explicit computation of the sky curvature are illustrated in Schwarzschild spacetime.
\end{abstract}


\maketitle

One of the reasons of the significance of conformal invariants in the study of the geometry of Lorentzian geometry and gravitation is the fact that the causal properties of a given spacetime depend solely on its conformal structure.     Conformal invariants which are polynomial on the metric tensor and its derivatives were thoroughly described by Fefferman and Clarke \cite{Fe85} culminating a long an arduous road started by Weyl and others more than half a century before (see, for instance \cite{Sz68} and references therein).        

The study of the space of light rays of a given spacetime $M$ offers an alternative way of looking at its causal structure, an idea that has its origin in R. Penrose's intuition and was brought to fruition by the hand of R. Low \cite{Lo88,Lo89,Lo90,Lo94,Lo00, Lo06} and the recent results obtained by the authors \cite{Ba14,Ba15}.   In particular, in \cite{Ba15} it was proved that the causal structure of a strongly causal sky separating spacetime is determined by a partial order relation on its space of skies and that the conformal class of the metric structure is completely determined by its infinitesimal sky structure \cite{Ba14}.    All these observations bring us the question if a search for conformal invariants program similar to that developed in the setting of the spacetime itself, could be started working on the space of light rays itself $\mathcal{N}$, the ``causal dual'' of $M$.      

The present paper offers a preliminary partial answer to this question by constructing a new conformal invariant, called the sky conformal invariant, or the sky-invariant for short, which is constructed using the geometrical ingredients available on the space $\mathcal{N}$ of light rays of $M$.      The main notion behind it is to construct a ``Ricci''-like tensor on the tangent spaces to the skies of events on $M$.   Such construction is inspired on the observation that for any parametrisation of a light ray $\Gamma$ there is a metric $g$ that makes it $g$-geodesic and which is unique up to multiplicative factor along the geodesic, hence for light rays a natural conformal covariant derivative can be defined, the so called Fermi-Walker connection.  Analyzing the structure of such conformal covariant derivative it will be shown that there is a canonical conformal invariant parametrisation obtained fromt the canonical projective parametrisation derived from the vanishing of the natural Ricci tensor associated to it. Such conformal covariant derivative, combined with the conformal parametrisation of light rays, allows to define a natural $(1,1)$-tensor $C_x \colon TS(x) \to TS(x)$ on an open set of the tangent bundle to the sky $S(x)$,  for any $x \in M$ that will be called the sky curvature. The principal elements of such tensor are the new conformal invariants we exhibit in this paper.    

It is relevant to point out here that this approach to construct the new conformal invariants is reminiscent of work done by Agrachev and collaborators on Jacobi curves and curvature invariants \cite{Ag02,Ag08,Ag15}.  Their relation to the present work will be discussed elsewhere.

The actual computation of these new invariants poses, in principle, significant difficulties as the explicit computations of the invariants would require the solution of non-linear differential equations, however, most of these complications can be circumvented and explicit formulae will be obtained for them that only require the computation of higher order derivatives of tensorial objects.   With the help of symbolic manipulation software, explicit expressions can be found and this procedure will be illustrated working them for Schwarzschild's spacetime.  

The paper will be organised as follows. In Sect. \ref{sec:lightrays} the main concepts and notations concerning the space of light rays will be succinctly reviewed.  In Sect. \ref{sec:conformal} the conformal geometry of a light ray and its canonical conformal covariant derivative will be discussed.   It will be shown that a canonical conformal parametrisation of light rays can be constructed and some of its properties will be discussed.   In Sect. \ref{sec:skycurvature} the natural curvature $(1,1)$ tensor on the tangent bundle of skies will be introduced and explicit formulas for its calculation will be presented.   Sect. \ref{sec:algorithm} will be devoted to sketch an algorithm to compute the sky-invariant of a given spacetime and, finally, Sect. \ref{sec:schwarzschild} will show the results obtained for Schwarzschild spacetime.


\section{The space of light rays $\mathcal{N}$}\label{sec:lightrays}

\subsection{The space of light rays}\label{sec:N}
Along the paper a spacetime $M$ will be considered to be a second countable paracompact $m$-dimensional smooth manifold carrying a conformal class  $\mathcal{C}$ of Lorentzian metrics of signature $(-+\cdots +)$ such that $M$ is time-orientable\footnote{A time-like vector field $T$ determining a time-orientation on $M$ will be fixed in what follows.} and strongly causal.   We will denote by $g$ a representative metric on $\mathcal{C}$ and $g_x(u,v)$, $u,v\in T_xM$, will denote the product defined by the metric $g$. 

Let $\mathcal{N}$ denote the space of unparametrized inextensible null geodesics, called in what follows \textit{light rays}, i.e., $\mathcal{N}$ is the space of equivalence classes of inextensible smooth null curves $\gamma \colon I \to M$, with $I$ an interval in $\mathbb{R}$, such that $\nabla_{\dot{\gamma}} \dot{\gamma} = 0$, $g \in \mathcal{C}$ and two such curves $\gamma$, $\tilde{\gamma}$ are equivalent if they define the same set in $M$, that is $\mathrm{Ran\,} (\gamma) = \mathrm{Ran\,} (\tilde{\gamma})$.   In what follows we will assume that the parametrisations $\gamma = \gamma(t)$ are future oriented, that is $g(\dot{\gamma}, T) < 0$,  The equivalence class containing the null geodesic $\gamma = \gamma (t)$, i.e., the light ray determined by $\gamma$  will be denoted by $\Gamma$. 

We will consider in what follows the fibre bundle $\mathbb{N}$ over $M$ consisting of nonzero null vectors, and its corresponding components $\mathbb{N}^\pm$ of future (past) null vectors.  If we denote $\mathbb{N}_x^+ = \{ v \in \mathbb{N}_x \mid v \neq 0, g_x(v, T(x)) < 0 \}$ and   $\mathbb{N}_x^- = \{ v \in \mathbb{N}_x \mid v \neq 0, g_x(v, T(x)) > 0 \}$, we have $\mathbb{N}^\pm = \bigcup_{x \in M} \mathbb{N}^\pm_x$ and $\mathbb{N} = \mathbb{N}^+ \cup \mathbb{N}^-$.  We will denote by $\pi \colon \mathbb{N} \to M$ the restriction of the canonical tangent bundle projection $TM \to M$ to $\mathbb{N}$.     

We will denote again by $\pi$ the canonical projection $\pi \colon \mathbb{PN}^+ \to M$, where $\mathbb{PN}^+$ denotes the quotient space of $\mathbb{N}^+$ by the action of the multiplicative group of positive real numbers $\mathbb{R}^+$ by scalar multiplication, i.e., $\alpha \in \mathbb{PN}^+$ denotes a ray $[u] = \{ \lambda u \mid u \in TM, \lambda > 0\}$ of tangent vectors containing $u$.   Notice that there is a canonical surjection 
\begin{equation}\label{eq:covering}
\sigma \colon \mathbb{PN}^+ \to \mathcal{N} \, , \qquad \sigma (\alpha ) = \Gamma_{\alpha} \, ,
\end{equation} 
which is the light ray containing  the unique future-oriented null geodesic $\gamma_u (t)$ such that $\gamma_u (0) = \pi (u)$, and $\dot{\gamma}_u (0) = u$, for some $g \in \mathcal{C}$.

\subsection{The smooth structure of $\mathcal{N}$}\label{atlas}
The space of light rays $\mathcal{N}$ can be equipped with the structure of a second countable paracompact smooth manifold of dimension $2m-3$ ($\dim M = m$), such that the map $\sigma$ becomes a submersion, in two different ways.   We will succinctly describe them in the following paragraphs (see, for instance, \cite{Ba22,Ba14} for details).  

First,  we can use the local structure of $M$, i.e., pick a representative $g \in \mathcal{C}$, and, because $(M,g)$ is strongly causal, given any event $x\in M$, there exists a globally hyperbolic neighbourhood $U_x$ of $x$ and a local smooth space-like Cauchy hypersurface $C_x \in U_x$ \cite{Mi08}.  We can take $U_x$ small enough such that it is contained in a local chart of $M$.  Hence we can define an atlas for $\mathcal{N}$ as follows, select for any event $x\in M$ a globally hyperbolic open neighbourhood $U_x$ as before with Cauchy hypersurface $C_x $.   We consider the restriction of the projective bundle $\mathbb{PN}^+$ to $C_x$ and we denote it by $\mathbb{PN}^+(C_x)$.   There is a natural embedding $i_x \colon \mathbb{PN}^+(C_x) \to \mathbb{PN}^+$.  Then the composition $\sigma \circ i_x \colon \mathbb{PN}^+(C_x) \to \mathcal{N}$ will provide the charts of the atlas we are looking for and the open sets $\mathcal{U}_x = \sigma \circ i_x ( \mathbb{PN}^+(C_x)) \subset \mathcal{N}$ will be the domains of the corresponding charts (see \cite[Sect. 2.3]{Ba14} for more details).

Alternatively, we can induce a smooth structure on $\mathcal{N}$ from the smooth structure of the bundle $\mathbb{N}^+$ by considering the foliation defined by the leaves of the integrable distribution generated by the vector fields $X_g$ and $\Delta$, where $X_g$ denotes the geodesic spray of a representative metric $g$ in the conformal class $\mathcal{C}$, and $\Delta$ is the dilation or Euler field along the fibres of $TM$.    Because $[X_g, \Delta] = X_g$, the distribution $D = \mathrm{span} \{ \Delta, X_g\}$ is integrable and denoting by $\mathcal{D}$ the corresponding foliation, we have that the space of leaves $\mathbb{N}^+ / \mathcal{D} \cong \mathcal{N}$.    If $M$ is strongly causal it can be shown that $\mathcal{D}$ is a regular foliation and the space of leaves inherits a smooth structure from $\mathbb{N}^+$.   Again, it is not hard to show that both smooth structures coincide (see also \cite{Ba16} for details).


\subsection{The tangent bundle $T\mathcal{N}$ and the contact structure of $\mathcal{N}$}\label{sec:tangent}
Let $\Gamma \colon (-\epsilon, \epsilon) \to \mathcal{N}$, $\Gamma = \Gamma (s)$, be a differentiable curve in $\mathcal{N}$, such that $\Gamma (0) = \Gamma$ and let $\chi (s,t) \colon (-\epsilon, \epsilon) \times I \to M$ be a geodesic variation by null geodesics of a parametrisation $\gamma (t)$ of the light ray $\Gamma$ and such that $\chi(s,t)$ is a parametrisation of the light ray $\Gamma (s)$.  In other words, $\chi$ is a smooth function such that $\chi (s,t) = \gamma_s (t)$ are null geodesics with respect to the metric $g \in \mathcal{C}$, $\gamma_0 (t)$ is a parametriwation of $\Gamma$, and $[\gamma_s] = \Gamma (s)$, where $[\gamma_s]$ denotes the unparametrized geodesic containing $\gamma_s$.  Then the vector field along $\gamma$ defined by $J(t) = \partial \chi (s,t)/\partial s \mid_{s = 0}$ is a Jacobi field.  The set of Jacobi fields along $\gamma (t)$ will be denoted by $\mathcal{J} (\gamma )$  and they satisfy the second order differential equation:
\begin{equation}\label{eq:jacobi}
J'' = R(\dot{\gamma}, J) \dot{\gamma} \, ,
\end{equation}
where $J' = \nabla_{\dot{\gamma}}J$ denotes the covariant derivative of $J$ along $\dot{\gamma}(t)$.  Notice that since the geodesic variation $\chi$ is by null geodesics, we have $g (J, \dot{\gamma} ) =$ constant, and we denote by $\mathcal{L}(\gamma)$ the linear space of Jacobi fields satisfying this property.

Equivalence classes of curves $\Gamma (s)$ possessing a first order contact define tangent vectors to $\mathcal{N}$ at $\gamma$, hence tangent vectors at $\gamma$ correspond to equivalence classes of Jacobi fields with respect to the equivalence relation defined by reparametrisation of the geodesic variation $\chi$.   Such reparametrisations will correspond to Jacobi fields of the form $(at + b) \dot{\gamma}(t)$, then there is a canonical projection $\mathcal{L}(\gamma) \to T_\Gamma \mathcal{N}$, mapping each Jacobi field $J$ into a tangent vector $\langle J \rangle = J \,\, \mathrm{mod} \dot{\gamma}$ whose kernel is given by Jacobi fields proportional to $\dot{\gamma}$. 

There is a canonical contact structure on $\mathcal{N}$ defined by the maximally non-integrable hyperplane distribution 
$\mathcal{H}_\Gamma \subset T_\Gamma \mathcal{N}$ formed by the vectors orthogonal to their supporting light ray, i.e., 
\begin{equation}\label{contact}
\mathcal{H}_\Gamma = \{ \langle J \rangle \in T_\Gamma \mathcal{N} \mid g( J, \dot{\gamma} ) = 0  \} \, .
\end{equation}
The contact structure $\mathcal{H}$ does not depend on the representative metric $g \in \mathcal{C}$, used to define the representative $J$ chosen for the tangent vector, or the parametrisation $\gamma(t)$ we choose for the light ray $\gamma$ \cite{Ba14,Ba15}.    


\subsection{Skies}\label{sec:skies}
The congruence of light rays $S\left(p\right)\subset \mathcal{N}$ passing by a point $p\in M$ characterises the point $p$. 
We will call the \emph{sky of $p$} the set:
\[
S\left(p\right)=\left\{ \Gamma\in\mathcal{N} \mid p\in\Gamma\subset M \right\}  \, .
\]
For a given $p\in M$, we can identify $S\left(p\right)$ with $\mathbb{PN}^+_p$, and since the fibre $\mathbb{PN}^+_p$ is diffeomorphic to the sphere $\mathbb{S}^{m-2}$ then we have that every sky $S\left(p\right)$ is a smooth submanifold of $\mathcal{N}$ diffeomorphic to $\mathbb{S}^{m-2}$.

By the description of $T_{\Gamma}\mathcal{N}$ in section \ref{sec:tangent}, given $x\in M$ and $\Gamma \in S\left(x\right)$ and $\gamma(t)$ a parametrisation of $\Gamma$  such that  $\gamma \left(s_{0}\right) =p$, then for $\langle J \rangle\in T_{\Gamma}S\left(p\right) \subset T_{\Gamma}\mathcal{N}$, since $g \left( J,\dot{\gamma}\right)$ is constant, and $J$ can be defined by a null geodesic variation with $p$ as a fixed point, then $J\left( s_{0}\right) =0\left( \mathrm{mod} \dot{\gamma}\right) $, hence:
\begin{equation}\label{eq-sky-tangent}
T_{\Gamma} S(p) =\{\langle J \rangle\in T_{\Gamma }\mathcal{N} \mid J\left( s_{0}\right) =0\left(\mathrm{mod}\dot{\gamma}\right) \}  \, ,
\end{equation}
and since $g \left( J\left(s_0\right),\dot{\gamma}\left(s_0\right)\right)=0$ then we have that $g \left( J, \dot{\gamma}\right)=0$, therefore $T_{\Gamma }S(p) \subset \mathcal{H}_{\Gamma }$.


\subsection{Sky conformal invariants}  It would be possible, in principle, to use the space of light rays $\mathcal{N}$ to construct conformal invariants as follows.    If $\phi \colon (M_1, \mathcal{C}_1) \to (M_2, \mathcal{C}_2)$ is a conformal diffeomorphism, that is, $\phi$ is a diffeomorphism such that, $\phi^*g_2 \in \mathcal{C}_1$ for any $g_2 \in \mathcal{C}_2$, then it maps light rays $\Gamma_1 \in \mathcal{N}_1$ into light rays $\phi (\Gamma_1)\in \mathcal{N}_2$.  Then a family of smooth maps $\mathscr{S}\colon \mathcal{N} \to \mathbb{R}$, such that $\mathscr{S}_2 (\phi(\Gamma_1)) = \mathscr{S}_1(\Gamma_1)$, for all $\Gamma_1 \in \mathcal{N}_1$, would be a conformal invariant.    However this is not a good idea as it is often hard to describe the space of light rays $\mathcal{N}$ explicitly.    We can make this notion slightly more general and at the same time much more useful by considering the bundle $\mathcal{S}$ of skies over $M$, i.e., the disjoint union of all skies $S(x)$, $x \in M$, instead.     A convenient way of understanding the space $\mathcal{S}$ is as a ``blowing up'' of the spacetime $M$ where every point $x \in M$ is replaced by the congruence of light rays passing through it, that is (see \cite{Ba17,Ba18} for details):
$$
\mathcal{S} = \{ (x,\Gamma) \in M \times \mathcal{N} \mid \Gamma \in S(x) \} \, .
$$
There is a natural projection $\rho_1 \colon \mathcal{S} \to M$, given by $\rho_1 (x,\Gamma) = x$.  Note that $\rho_1^{-1}(x)$ is the sky $S(x)$.
The projection $\rho_2 \colon \mathcal{S} \to \mathcal{N}$, $\rho_2 (x,\Gamma) = \Gamma$, shows that the space $\mathcal{N}$ ``unfolds'' the space of light rays, in fact the fibre $\rho_2^{-1}(\Gamma)$ is the one-dimensional submanifold of $M$ given by the light ray itself.

It is noticeable that the bundle of spheres $\rho_1 \colon \mathcal{S} \to M$, can be naturally identified with the bundle $\pi \colon \mathbb{P}\mathbb{N}^+ \to M$, as any sky $S(x)$ is naturally identified with the fibre $\mathbb{P}\mathbb{N}^+_x$ (see above, Sect. \ref{sec:skies}). 
Note that the space $\mathbb{P}\mathbb{N}^+$ is a fibration over $\mathcal{N}$, Eq. (\ref{eq:covering}), as well as over $M$, thus on one hand ``unfolds'' the space of light rays $\mathcal{N}$, but its bundle structure over $M$ makes it much more suitable to construct conformal invariants. 

A \textit{sky conformal invariant} (or just \textit{sky invariant}) is a family of maps $\kappa \colon \mathbb{P}\mathbb{N}^+ \to \mathbb{R}$, such that:
$$
\kappa_2 (\phi_*(\alpha_1)) = \kappa_1 (\alpha_1) \, ,
$$ 
for all $\alpha_1 \in \mathbb{P}\mathbb{N}_1^+$, $\phi \colon (M_1, \mathcal{C}_1) \to (M_2, \mathcal{C}_2)$ a conformal diffeomorphism, $\phi_* \colon \mathbb{P}\mathbb{N}_1^+ \to \mathbb{P}\mathbb{N}_2^+$, the map induced by the tangent map $\phi_*\colon TM_1 \to TM_2$, and $\kappa_a \colon \mathbb{P}\mathbb{N}_a^+ \to \mathbb{R}$, $a = 1,2$, the corresponding maps. 

There is a natural way of constructing sky conformal invariants associated to the geometrical structure of light rays.  Consider a light ray $\Gamma$ and a parametrisation $\gamma(t)$ of it and suppose that we have a map $\kappa_\gamma (t)$ such that $\kappa_{\phi\circ \gamma}(t) = \kappa_\gamma (t)$, and $\kappa_\gamma (t+a) = \kappa_\gamma (t)$ for all $a$; we will say that $\kappa_\gamma$ is a \textit{parametric conformal invariant}.   If, in addition, $\kappa_\gamma (t)$ is invariant under reparametrisations of $\gamma (t)$, that is, if $\bar{\gamma} (\bar{t}) = \gamma (\psi (\bar{t})) = \gamma (t)$, is another parametrisation of $\Gamma$ (that is, $t = \psi (\bar{t})$, is a regular, future oriented, reparametrisation of  $t$), then $\kappa_{\bar{\gamma}}(\bar{t}) = \kappa_\gamma (t)$, we will say that $\kappa_\gamma$ is an \textit{absolute conformal invariant}.  

If $\kappa_\gamma$ is an absolute conformal invariant it is obvious that it defines a map (denoted with the same symbol) $\kappa \colon  \mathbb{P}\mathbb{N}^+ \to \mathbb{R}$, by means of $\kappa (\alpha) = \kappa_{\gamma_\alpha} (0)$, where $\gamma_\alpha (t)$ is the null $g$-geodesic such that $\gamma_\alpha (0) = \pi (\alpha) = p$, and $\dot{\gamma}_\alpha (0) = u$, if $\alpha = [u]$ (notice that two geodesics satisfying the previous conditions correspond to two different parametrisations of the same light ray $\Gamma_\alpha$) and $\kappa$ is a sky conformal invariant.   Hence, absolute conformal invariants $\kappa_\gamma$ define sky invariants $\kappa$.    This constitutes the main idea behind our strategy to construct sky conformal invariants: exploit the conformal geometry of light rays to construct absolute conformal invariants in the previous sense, that will give rise to sky conformal invariants.  The details of the construction of absolute conformal invariants will be discussed in the following sections.


\section{The conformal geometry of a light ray}\label{sec:conformal}


\subsection{Metric adapted to a parametrised light ray}

Note that if $\Gamma$ is a light-ray and we parametrise it using a regular future oriented parameter $t$, i.e., $\Gamma$ is the image of a parametrised curve $\gamma = \gamma (t)$, with $\dot{\gamma}(t) \neq 0$, $g( T(t), \dot{\gamma}(t) ) < 0$, for all $t$,  then there is a metric $g\in \mathcal{C}$ such that $\gamma (t)$ is a null geodesic for $g$.    Finally, if a change of parameter $t = \psi (\bar{t})$ is performed on the curve $\gamma(t)$, i.e., $\bar{\gamma}(\bar{t}) = \gamma (\psi (\bar{t}))$, then a simple computation shows:
\begin{equation}
g_{\bar{\gamma}(t)} (\bar{t}) = e^{2\varphi (\bar{t})} g_{\gamma(t)} (\psi ( \bar{t})) \, , 
\end{equation}
where,
\begin{equation}\label{eq:change}
 \varphi(\bar{t}) = - \frac{1}{2} \log \left( \frac{d\psi}{d\bar{t}} \right) \, ,
\end{equation}
or, in other words, $g_{\bar{\gamma}} (\bar{t}) = (d\psi / d\bar{t})^{-1} g_\gamma (\psi (\bar{t}))$, with $g_\gamma (t) := g_{\gamma(t)}$. 


\subsection{Conformal covariant derivative}

We are ready now to define the conformal covariant derivative along a curve, that is, fix a metric $g \in \mathcal{C}$, then consider a curve $\gamma (t)$, $t \in I$, (not necessarily a $g$-geodesic), then we define a linear map $\nabla /dt$ from the space of vector fields along the curve $\gamma$, $A(\gamma(t)) \in T_{\gamma (t)} M$, denoted in what follows $\mathfrak{X}_\gamma$, $t \in I$, by:
$$
\frac{\nabla A}{dt} = \nabla_{\dot{\gamma}}A \, , 
$$
with $\nabla$ the Levi-Civita connection defined by the metric $g$ and $\dot{\gamma}$ the tangent vector to $\gamma (t)$.   If $\gamma$ is $g$-pregeodesic, the map $\nabla/dt$ maps the space of vectors orthogonal to $\dot{\gamma}$ into itself, that is, if we denote by $\mathfrak{X}_\gamma^\perp$ the set of vectors $A\in \mathfrak{X}_\gamma$ such that $g( A,  \dot{\gamma} ) = 0$, then: $\nabla A /dt \in \mathfrak{X}^\perp_\gamma$ for each $A \in \mathfrak{X}_\gamma^\perp$.  

If $\gamma(t)$ parametrises a light ray $\Gamma$, it satisfies $g( \dot{\gamma}, \dot{\gamma} ) = 0$, then it is natural to consider the quotient space $\langle \mathfrak{X}_\gamma \rangle$ of vectors along $\gamma$ module $\dot{\gamma}$, that is $\langle A \rangle \in \langle \mathfrak{X}_\gamma \rangle$, denotes the set $A + f \gamma'$, with $f$ an arbitrary function and the subspace $\langle \mathfrak{X}_\gamma^\perp \rangle$ of equivalence classes of vector orthogonal to $\dot{\gamma}$.    Then, the map $\nabla /dt $ defined previously induces a map, denoted with the same symbol, $\nabla /dt \colon \langle \mathfrak{X}_\gamma^\perp \rangle \to \langle \mathfrak{X}_\gamma^\perp \rangle$, as:
$$
\frac{\nabla \langle A \rangle}{dt} = \left\langle \frac{\nabla A}{dt} \right\rangle \, .
$$
Note that the map $\nabla /dt$ is well defined on the quotient space and as it can be shown easily, it does not depend on the chosen metric $g \in \mathcal{C}$.     

We can summarise the previous discussion by saying that given a light ray $\Gamma$, and we choose a parametrisation $\gamma = \gamma(t)$, the conformal structure $\mathcal{C}$ induces a unique covariant derivative $\nabla^\gamma /dt$ on the quotient space $\langle \mathfrak{X}_\gamma^\perp \rangle$, that coincides with the covariant derivative determined by any metric $g\in \mathcal{C}$ such that $\gamma (t)$ is null $g$-geodesic.  We will call such map, the conformal covariant derivative on $\gamma = \gamma(t)$.

Notice that $\nabla^\gamma/dt$ depends just on the values of $g$ on $\gamma$ and it satisfies the following change of parameter formula:
$$
\frac{\nabla^{\bar{\gamma}} \langle \bar{A} \rangle}{d\bar{t}} = \frac{\nabla^\gamma \langle A \rangle}{d t} \frac{d\psi}{d\bar{t}} + \frac{d\varphi}{d\bar{t}} \langle \bar{A} \rangle \, ,
$$
with $t = \psi (\bar{t})$ a regular change of parameter, $\bar{\gamma}(\bar{t}) = \gamma (\psi (\bar{t})) = \gamma (t)$, $\bar{A} \in \langle \mathfrak{X}_{\bar{\gamma}}^\perp \rangle$, $A(\gamma(t)) = \bar{A}(\bar{\gamma}(\bar{t})) \in \langle \mathfrak{X}_{\gamma}^\perp \rangle$, and the function $\varphi (\bar{t})$ is given by Eq. (\ref{eq:change}). In particular, the reparametrisation $t = \psi (\bar{t})$ is a translation of parameters iff it induces the same conformal derivative in the light ray $\Gamma$.
Note that we can repeat the argument again, and define the conformal derivative $\nabla /dt \langle \nabla A/dt
 \rangle = \langle \nabla^2 /dt^2 A \rangle \in \langle \mathfrak{X}_\gamma^\perp \rangle$, for any vector field $\langle A \rangle \in  \langle \mathfrak{X}_\gamma^\perp \rangle$.
 

\section{Sky curvature and conformal parametrisations}\label{sec:skycurvature}


\subsection{Parametric sky curvatures}\label{sec:parametric}
Given a light ray $\Gamma$, and a $\gamma(t)$ a parametrisation of $\Gamma$, we can define a linear map $R_\gamma \colon T_\Gamma \mathcal{N} \to T_\Gamma \mathcal{N}$, given by $\langle J \rangle \mapsto \langle J'' \rangle$, where $J'' = \nabla^2 J /dt^2$, with $g\in \mathcal{C}$.   In other words, $R_\gamma \langle J \rangle = \nabla^2 \langle J \rangle /dt^2$, with $\nabla/dt$ the conformal covariant derivative along the parametrisation $\gamma(t)$ of $\Gamma$.   

Because $\langle J \rangle$ is a quotient Jacobi field, i.e., $J$ satisfies Jacobi equation (\ref{eq:jacobi}), then $R_\gamma \langle J \rangle = \langle R(\dot{\gamma}, J) \dot{\gamma} \rangle$.  Even more, because $R_\gamma$ is defined pointwise, we can define for each $t$, the linear map:
\begin{equation}\label{eq:Rgamma}
R_\gamma(t) \colon \langle \dot{\gamma}(t)^\perp \rangle  \to \langle \dot{\gamma}(t)^\perp \rangle \, , \qquad R_\gamma(t) \langle J \rangle = \langle R(\dot{\gamma}(t), J(t)) \dot{\gamma}(t) \rangle \,   .
\end{equation}
We realise that the tangent space $T_\Gamma S(p)$ to the sky $S(p)$, $p \in M$, at the light ray $\Gamma \in S(p)$ can be identified naturally with the quotient space $T_p\Gamma^\perp/T_p\Gamma = \langle \dot{\gamma}(0)^\perp \rangle$, using the induced linear map on $T_p\Gamma^\perp$ that assigns to each tangent vector $\xi \in T_p\Gamma^\perp$, the unique Jacobi field $J_\xi(t)$ along a parametrisation $\gamma (t)$ of $\Gamma$ as a $g$-geodesic, such that $\gamma (0) = p$, $J_\xi(0) = 0$ and $J_\xi'(0) = \xi$.    It is noticeable that the map $\xi \mapsto J_\xi$, induces a map $\langle \xi \rangle \mapsto \langle J_\xi \rangle$ among the quotient classes and it does not depend on the metric $g \in \mathcal{C}$ chosen for the parametrisation of $\Gamma$.  The map is clearly injective and it is surjective because of the description of the tangent space to the sky $S(p)$, Eq. (\ref{eq-sky-tangent}). Thus we can state:

\begin{theorem}  Let $\Gamma$ be a light ray in $\mathcal{N}$, and $p \in \Gamma$ a fixed point in the light ray.   
Then there is a canonical isomorphism between the linear space $T_p\Gamma^\perp /T_p\Gamma$ and the tangent space $T_\Gamma S(p)$ of the sky $S(p) \subset \mathcal{N}$ at $\Gamma$.
\end{theorem}

We may use the maps $R_\gamma (t)$, Eq. (\ref{eq:Rgamma}), to define an endomorphism $C_{\gamma,p}$ of the sky $T_\Gamma S(p) \subset T_\Gamma \mathcal{N}$, as $C_{\gamma,p} \langle J \rangle = \langle R_\gamma (0) \rangle$, where $\langle J \rangle$ is a quotient Jacobi field along $\Gamma$. 
Note that because of the previous observations, $C_{\gamma,p}$ is a linear map $C_{\gamma,p} \colon \langle T_p\Gamma^\perp \rangle \to \langle T_p\Gamma^\perp \rangle$, given by:
\begin{equation}\label{eq:sky_curvature}
C_\gamma (p) \langle J \rangle  = \langle R(\dot{\gamma}, J) \dot{\gamma}(0) \rangle \, , \qquad \forall \langle J \rangle \in T_\Gamma S(p) \, ,
 \end{equation}
which is invariant under translations and satisfies:
$$
C_{\phi \circ \gamma, \phi (p)}  = C_{\gamma, p} \, ,
$$
for any conformal diffeomorphism $\phi$.  Let $\alpha \in \mathbb{P}\mathbb{N}^+_p$, and let $\Gamma_\alpha \in \mathcal{N}$, be the unique light ray such that $T_p \Gamma_\alpha = \alpha$, then, because of the previous Theorem, $\langle T_p \Gamma_\alpha^\perp \rangle \cong \alpha^\perp / \alpha \cong T_\Gamma S(p)$.    Hence the sky-curvature $C_\gamma$ assigns to each point $p \in \Gamma$ a linear map $C_\gamma (p) \colon \alpha^\perp/\alpha  \to \alpha^\perp/ \alpha$, i.e. $C_\gamma (p) \langle J \rangle  = \langle R(\dot{\gamma_\alpha}, J) \dot{\gamma}_\alpha (0) \rangle$.      

On the other hand, the coefficients $\kappa_k(t)$ of the characteristic polynomial $p_\gamma (s) = \det (R_\gamma (t) - sI)$, allows us to define a family of functions $\kappa_{\gamma,k}(t)$ on each light ray $\Gamma$, as the $k$th coefficient of the characteristic polynomial of  $R_{\gamma}(t)$, and $\gamma (t)$ a parametrisation of $\Gamma$.  We will call such functions parametric curvatures and they are invariant under translations of the parameter.   Again because of the conformal nature of the covariant derivative along $\Gamma$, it is satisfied that $\kappa_{\gamma,k} (t) =\kappa_{\phi\circ \gamma,k}(t)$ for any conformal diffeomorphism $\phi$.  Thus, the functions $\kappa_{\gamma,k} (t)$ define a family of parametric conformal invariants that will be called parametric sky curvatures.

In particular we can define two parametric curvatures: $\rho_\gamma (t) = \det R_\gamma(t) =  \kappa_{\gamma,m-2}(t)$, and $\delta_\gamma (t) = \mathrm{Tr} (R_{\gamma} (t)) = (-1)^{m-3}\kappa_{\gamma,1}(t)$, which are the only scalars associated to the linear map $R_\gamma(t)$ in four dimensions (note the if $m = 4$, the dimension of skies is 2). 

In principle, the parametric curvatures $\kappa_{\gamma,k}$ are not absolute conformal invariants, that is, they are not invariant under reparametrisations of $\Gamma$.   The reason for this is that if $t = \psi (\bar{t})$, we get (after a simple but subtle computation):
\begin{equation}\label{eq:curvature_change}
R_{\bar{\gamma}}(\bar{t}) \langle \bar{J} \rangle = \left( \frac{d\psi}{d\bar{t}} \right)^2 R_\gamma (\psi (\bar{t})) \langle J \rangle + \left\{ \frac{d^2\varphi}{d\bar{t}} + \left( \frac{d\varphi}{d\bar{t}}\right)^2 \right\} \langle J \rangle \, , 
\end{equation}
with $\bar{J}(\bar{t}) = J (\psi (t))$, and $\varphi$ given by Eq. (\ref{eq:change}), and they will not allow us to define directly a family of sky invariants $\kappa_k \colon  \mathbb{P}\mathbb{N}^+ \to \mathbb{R}$.  Moreover, computing them starting from an arbitrary parametrisation of a light ray $\Gamma$ is already hard as it involves determining a metric $g \in \mathcal{C}$ that makes them $g$-geodesic.   We will see in the following paragraphs, that both difficulties can be solved simultaneously by introducing a particular family of parametrisations that would allow us to turn the parametric conformal invariants $\kappa_{\gamma,k}$ into absolute ones and at the same time will provide an algorithmic way to compute them.


\subsection{Conformal parameter of a light ray}

We will denote by $\mathrm{Ric}_\gamma$ the trace of the curvature tensor $R_\gamma$ associated to the conformal covariant derivative $\nabla/dt$ along the light ray $\Gamma$ with parametrisation $\gamma (t)$.  Then, because of Eq. (\ref{eq:curvature_change}), we get:
\begin{equation}\label{eq:ric_change}
\mathrm{Ric}_{\bar{\gamma}} = \left(\frac{d\psi}{d\bar{t}}\right)^2 \mathrm{Ric}_\gamma + (m-2) \left\{ \frac{d^2\varphi}{d\bar{t}} + \left( \frac{d\varphi}{d\bar{t}}\right)^2 \right\} \, ,
\end{equation}
with $\varphi = - \frac{1}{2} \log d\psi /dt$, and $t = \psi (\bar{t})$ the change of parameters.   A simple computation shows that:
\begin{equation}\label{eq:schwarzian}
\frac{d^2\varphi}{d\bar{t}} + \left( \frac{d\varphi}{d\bar{t}}\right)^2 = - \frac{1}{2} \frac{\psi'''}{\psi'} + \frac{3}{4} \frac{(\psi'')^2}{(\psi')^2} = - \mathbb{S}(\psi) \, , 
\end{equation}
where the function $\mathbb{S}(\psi)$ is called the Schwartzian derivative of $\psi$.    It is well known that functions $\psi$ such that $\mathbb{S}(\psi) = 0$, must have the form:
\begin{equation}\label{eq:projective}
\psi (\bar{t}) = \frac{a \bar{t} + b}{c \bar{t} + d} \, , \qquad \det \left(\begin{array}{cc} a & b \\ c & d\end{array} \right) \neq 0 \, .
\end{equation}
We will say that a parametrisation $\gamma_P (t)$ of the light ray $\Gamma$ is projective if $\mathrm{Ric}_{\gamma_P} = 0$. In such case, we will say that $t$ is a projective parameter for $\Gamma$.    Observe that if $\bar{\psi} = \psi^{-1}$, then,
$$
\mathbb{S}(\bar{\psi})\mid_t = - \left( \frac{d\bar{\psi}}{dt}\right)^2 \mathbb{S}(\psi)\mid_{\bar{\psi}(t)} \, .
$$
If $\bar{t} = \psi (t)$ is a projective parameter, then $\mathrm{Ric}_{\bar{\gamma}} = 0$, and, because of Eqs. (\ref{eq:ric_change},\ref{eq:schwarzian}), we get:
\begin{equation}\label{eq:SRic}
\mathbb{S}(\bar{\psi}) = \frac{1}{2- m} \mathrm{Ric}_\gamma \, ,
\end{equation}
then, the solutions of the third order differential equation (\ref{eq:SRic}) permit to obtain the projective parameter $\bar{t}$ from any given parameter $t$, and, in such case Eq. (\ref{eq:curvature_change}) becomes:
\begin{equation}\label{eq:R_conformal}
R_{\bar{\gamma}}(\bar{t}) = \left( \frac{d\psi}{d\bar{t}} \right)^2 \left[ R_\gamma (\psi (\bar{t})) -   \frac{1}{m - 2} \mathrm{Ric}_\gamma(\psi (\bar{t})) \mathrm{Id}  \right] \, .
\end{equation}
Note that if $\gamma_P(t)$ and $\gamma_P(t')$ are two projective parametrisations of $\Gamma$, then because of (\ref{eq:projective}), $t' = (at + b)/(ct + d)$, $ad-bc > 0$.

We can define the length of a segment $\Gamma_0$ of the light ray $\Gamma$ as follows.  Let $\gamma_P(t)$ be a projective parametrisation of $\Gamma$ and $\Gamma_0 = \gamma_P([a,b])$ be a segment of $\Gamma$, then define:
\begin{equation}\label{eq:conformal_par}
L(\Gamma_0) = \int_a^b \zeta_{\gamma_P}(t) dt \, , \qquad \zeta_{\gamma_P}(t) = \sqrt[2(m-2)]{|\det R_{\gamma_P(t)}|} \, .
\end{equation}
The previous definition is independent of the chosen projective parametrisation.   Indeed, if $\bar{t}$, ($t = \psi (\bar{t})$) is another projective parameter, then $\mathbb{S}(\psi) = 0$, because of (\ref{eq:ric_change}), we get:
$$
R_{\bar{\gamma}_P} =  \left(\frac{d\psi}{d\bar{t}}\right)^2 R_\gamma \, ,
$$
and we conclude:
$$
\zeta_{\bar{\gamma}_P} = \sqrt[2(m-2)]{|\det R_{\bar{\gamma}_P(\bar{t})}|} = \sqrt[2(m-2)]{ \left(\frac{d\psi}{d\bar{t}}\right)^{2(m-2)}|\det R_{\gamma_P(t)}|} = \frac{d\psi}{d\bar{t}} \zeta_{\gamma_P} (\psi (\bar{t})) \, .
$$
Then, the change of variables formula gives us:
$$
\int_a^b \zeta_{\gamma_P}(t) dt = \int_{\bar{a}}^{\bar{b}} \zeta_{\bar{\gamma}_P} (\bar{t}) d\bar{t} \, ,
$$
and the length $L(\Gamma_0)$ does not depend on the chosen projective parametrisation.  Moreover, the length $L$ thus defined is a conformal invariant, that is if $\phi \colon (M, \mathcal{C}) \to (M', \mathcal{C}')$ is a conformal diffeormophism, it transforms light rays into light rays and the length $L(\Gamma_0)$ of the segment of the light ray $\Gamma \subset M$, coincides with the length $L(\phi(\Gamma_0))$.   

We may also say that there is a 1-form $ds_\Gamma$ defined on the light ray $\Gamma$, that can be written as $ds_\Gamma (t) = \zeta_{\gamma_P}(t) dt$ with respect to any projective parametrisation, such that:
$$
L(\Gamma_0) = \int_{\Gamma_0} ds_\Gamma \, .
$$
Notice that the vanishing of $R_\gamma$ implies the vanishing of $ds_\Gamma$, hence, for conformally flat spaces, $ds_\Gamma = 0$.

In the particular instance $\dim M = 4$, if $\gamma_P(t)$ is a projective parametrisation of $\Gamma$, then $\mathrm{Ric}_{\gamma_P} = 0$, and provided that $R_{\gamma_P}(t_0) \neq 0$, then necessarily $|\det R_{\gamma_P}(t_0)| > 0$.   Then in a neighborhood $I$ of $t_0$, $t \in I$, $\zeta_{\gamma_P}(t) > 0$, and defining:
$$
s = \psi^{-1}(t) = \int_{t_0}^t \zeta_{\gamma_P}(\tau) d \tau \, ,
$$
then, $ds/dt = \zeta_{\gamma_P}(t) > 0$, and $t = \psi (s)$ is the parametrisation by conformal arc of $\Gamma$, that is $\gamma_C(s) =\gamma_P(\psi (s))$.   Notice that the conformal parameter $s$ is defined up to translations.

We will conclude this section providing a general expression for the differential of conformal arc $ds_\Gamma$.   That is, let $\gamma(t)$ be an arbitrary parametrisation of $\Gamma$, then because of Eq. (\ref{eq:R_conformal}):
\begin{equation}\label{eq:diff_conformal}
ds_\Gamma = \sqrt[2(m-2)] {\left| \det\left( R_\gamma (t) - \frac{1}{m-2} \mathrm{Ric}_\gamma (t) \cdot \mathrm{Id} \right)  \right| } dt \, .
\end{equation}


\subsection{The sky-curvature and the absolute sky curvatures $\rho$ and $\delta$}

Let $p \in M$  be such that $ds_\Gamma \neq 0$ for all $\Gamma \in S(p)$ (obviously, if $ds_\Gamma = 0$ the sky curvature cannot be defined).  Then, let $\gamma_C(s)$ be the parametrisation of $\Gamma$ by the conformal arc parameter $s$ with $\gamma_C(0) = p$.   

Let us consider again the parametric conformal curvatures $\kappa_{\gamma,k}$ defined in Sect. \ref{sec:parametric}.   We associate to them the corresponding absolute curvatures $\kappa_{\Gamma,k} \colon \Gamma \to \mathbb{R}$, defined as:
$$
\kappa_{\Gamma,k} (\gamma_C(s)) = \kappa_{\gamma_C,k}(s) \, ,
$$
where $\gamma_C(s)$ is the parametrisation of $\Gamma$ by the conformal parameter $s$ (which is defined provided that $\zeta_{\gamma_P}(t) \neq 0$ on $\Gamma$).    The absolute curvatures $\kappa_{\Gamma,k}$ can be written as maps $\kappa_k \colon \mathbb{P}\mathbb{N}^+ \to \mathbb{R}$, as discussed in Sect. \ref{sec:parametric}.   

We conclude then, that the linear map $C_\Gamma (p) = R_{\gamma_C}(0) \colon T_\Gamma S(p) \to T_\Gamma S(p)$, is a sky conformal invariant that will be called the sky-curvature.

In the particular instance of $\dim M = 4$, skies are two-dimensional and the sky-curvature $C_\Gamma$, assigns to any $\alpha \in  \mathbb{P}\mathbb{N}^+$, a linear map from the 2-dimensional space $T_{\Gamma_a}S(p)$ into itself, where $\pi (\alpha ) =p$ and $\Gamma_\alpha$ is, as usual, the only light ray such that $p \in \Gamma_\alpha$ and $T_p\Gamma_\alpha = \alpha$. Then, there are just two absolute conformal curvatures $\rho (\alpha) = \mathrm{Tr\,} C_{\Gamma_\alpha}(0)$, and $\delta (\alpha) = \det C_{\Gamma_\alpha (0)}$.

We will finish the discussion of the properties of the sky-curvature $C$ by providing an explicit expression for it using Eqs. (\ref{eq:ric_change}) and (\ref{eq:conformal_par}) that will be extremely useful for its computation (see below, Sect. \ref{sec:algorithm}):

\begin{theorem} Let $g \in \mathcal{C}$, $\alpha \in \mathbb{P}\mathbb{N}^+_p$, and $\gamma = \gamma_\alpha$, be a null $g$-geodesic such that $\gamma_\alpha (0) = p$, $\dot{\gamma}_\alpha (0) \in \alpha$. Then the sky-curvature $C_{\Gamma_\alpha}$ as an endormorphism of $\alpha^\perp/\alpha$, is given by:
\begin{equation}\label{eq:Cp}
C_{\Gamma_\alpha} = \frac{1}{\zeta_\gamma^2} \left[ R_{\gamma}(0) + \frac{1}{2} \left( \frac{\zeta''_{\gamma}}{\zeta_{\gamma}} - \frac{3}{2} \left(\frac{\zeta_\gamma'}{\zeta_\gamma} \right)^2 \right) \mathrm{Id} \right] \, ,
\end{equation}
with $\zeta_\gamma =  \sqrt[2(m-2)] {\left| \det\left( R_\gamma (t) - \frac{1}{m-2} \mathrm{Ric}_\gamma (t) \, \mathrm{Id} \right)  \right|}$, or, alternatively, using the function $D_{\alpha}(t) = \det \left( R_{\gamma}(t) - \frac{1}{m-2} \mathrm{Ric\,}_{\gamma}(t)  \, \mathrm{Id} \right)$, we can rewrite Eq. (\ref{eq:Cp}) as: 
\begin{equation}\label{eq:CD}
C_{\Gamma_\alpha} = \frac{1}{\sqrt[m-2]{| D_\alpha |}} \left[ R_{\gamma}(0) + \frac{1}{4(m- 2)} \left( \frac{D''_{\alpha}}{D_{\alpha}} + \frac{7-4m}{4(m-2)} \left(\frac{D_\alpha'}{D_\alpha} \right)^2 \right) \mathrm{Id} \right] \, .
\end{equation}
\end{theorem}

\begin{proof} Formula (\ref{eq:Cp}) is a straightforward consequence of the definition of the sky-curvature  (\ref{eq:sky_curvature}) and the expression for the conformal parameter (\ref{eq:R_conformal}). 
\end{proof}

We will finish this preliminary study of the absolute conformal invariants associated to the sky-invariant by observing that in the 4-dimensional situation, the conformal curvatures $\rho$ and $\delta$ are functionally dependent.  Indeed, from the expression of the differential of the conformal parameter given by Eq. (\ref{eq:diff_conformal}), we get for $t$ a conformal parameter:
$$
\det \left( \delta (t) - \frac{1}{m-2} \rho(t) \right) = \epsilon = \pm 1 \, .
$$
Again in the case $m = 4$, we get, denoting $R_\Gamma (t) = (R_{ij})$, $i,j = 1,2$:
\begin{equation}\label{eq:epsilon}
\epsilon = \delta - (R_{11} + R_{22}) r = \delta - \frac{1}{4} \rho \, .
\end{equation}
Notice that the previous expression (\ref{eq:epsilon}) shows that the sign $\epsilon$ is a conformal invariant too.   The sky curvature matrix $R_\Gamma$ has characteristic polynomial $P(\lambda) = \lambda^2 - \rho \lambda + \delta$ with eigenvalues: $(\rho \pm \sqrt{-4\epsilon})/2$, hence if $\epsilon = -1$, there are two real eigenvalues:  $\rho/2 \pm 1$, that correspond to two different eigenvectors of the sky curvature and $\Gamma \in S(x_0)$ would be hyperbolic.   On the other hand if $\epsilon = + 1$, then the sky curvature has imaginary eigenvalues and $\Gamma$ would be elliptic.


\section{An algorithm to compute the sky curvature}\label{sec:algorithm}

Computing the sky-curvature of a given spacetime is a demanding problem.  In this section we will discuss an algorithm to do it that can be implemented on a symbolic manipulation software\footnote{In our case, it was implemented in Mapple\copyright .}.    In the following section, the algorithm is applied to the Schwarzchild spacetime.

\medskip

\textbf{Setting the problem:}

\begin{enumerate}

\item Given a spacetime of dimension 4 with conformal structure $(M,\mathcal{C})$, select a point $p \in M$.

\item Fix a metric $g\in \mathcal{C}$, and a $g$-orthonormal frame: $\boldsymbol{\varepsilon}=\left( \varepsilon_0, \varepsilon_1, \varepsilon_2, \varepsilon_3 \right)$, $\varepsilon_\mu \in T_pM$, $g(\varepsilon_0,\varepsilon_0) = -1$, $g(\varepsilon_i,\varepsilon_i) = 1$, $i=1,2,3$. Then we get:
$$
\mathbb{P}\mathbb{N}_p^+ = \left\{ \left[\boldsymbol{\varepsilon} \cdot \left( \begin{array}{c} 1 \\ \boldsymbol{\alpha} \end{array}\right) \right] \mid \boldsymbol{\alpha} = (\alpha_1,\alpha_2,\alpha_3)^T, \alpha_1^2 + \alpha_2^2 + \alpha_3^2 = 1 \right\} \cong S^2
$$
\item Define a coordinate system $(A,B)$ in $S(p) \cong \mathbb{P}\mathbb{N}^+$, $0\leq A < \pi$, $0 \leq B < 2\pi$, as:
\begin{equation}\label{eq:alphaAB}
\alpha = \alpha (A,B) =  \left[\boldsymbol{\varepsilon} \cdot \left( \begin{array}{c} 1 \\ \boldsymbol{\alpha}(A,B) \end{array} \right) \right] \, ,
\end{equation}
with 
\begin{equation}\label{eq:alpha}
\boldsymbol{\alpha}(A,B) = \left( \begin{array}{c} \sin A \cos B \\ \sin A \sin B \\ \cos A \end{array}\right)\, .
\end{equation}
\item The problem is to compute the matrix 
\begin{equation}\label{eq:C}
C = \left( \begin{array}{cc} C_A^A & C_A^B \\ C_B^A & C_B^B \end{array}\right) \, ,
\end{equation} 
of the components of the sky-curvature $C_p$ in $ \mathbb{P}\mathbb{N}^+$ in the coordinates $(A,B)$.  From $C$ we get the conformal invariants $\delta = \det C$, and $\rho = \mathrm{Tr\, } C$.

\end{enumerate}

\bigskip

\textbf{The algorithm}

\begin{enumerate}

\item For each $\alpha = \alpha (A,B)$ (\ref{eq:alphaAB}), construct the orthogonal basis in $T_pM$:
\begin{equation}\label{eq:epsilon2}
\boldsymbol{\varepsilon} (\alpha) = \boldsymbol{\varepsilon}\cdot \left( \begin{array}{cccc} 1 & 0 & 0 & 0 \\ 0 & \boldsymbol{\alpha} & \partial  \boldsymbol{\alpha} /\partial A & \partial  \boldsymbol{\alpha} /\partial B \end{array}\right) \, ,
\end{equation}
and write $\boldsymbol{\varepsilon} (\alpha ) = ( \mathbf{e}_0(\alpha), \mathbf{e}_1(\alpha),  \mathbf{e}_2(\alpha),  \mathbf{e}_3(\alpha))$.

\item  Consider linear coordinates $\left(\begin{array}{c} \lambda \\ \mu \end{array} \right)$ in $T_{\Gamma_\alpha} S(p) = \alpha^\perp /\alpha$, of the form:
\begin{equation}\label{eq:lambdamu}
\left(\begin{array}{c} \lambda \\ \mu \end{array} \right) \mapsto ( \mathbf{e}_2(\alpha),  \mathbf{e}_3(\alpha)) \left(\begin{array}{c} \lambda \\ \mu \end{array} \right) \quad \mathrm{mod\, } \alpha \, .
\end{equation}

\item Compute the functions $\mathbf{R}_{abc}^d (\alpha)$ such that:
$$
R( \mathbf{e}_a, \mathbf{e}_b)  \mathbf{e}_c = \mathbf{R}_{abc}^d  \mathbf{e}_d \, ,
$$
with $R$ the Riemann curvature tensor of $g$.
\item Let $\gamma_a(t)$ be the unique $g$-geodesic such that:
$$
\dot{\gamma}(0) =  \mathbf{e}_0 (\alpha) +  \mathbf{e}_1(\alpha) \, .
$$
Then, the matrix $R_{\gamma_a}(0)$ as an endomorphism of $\alpha^\perp /\alpha$ with respect to the basis $(\mathbf{e}_2,\mathbf{e}_3) \, \mathrm{mod \, } \alpha$, will have associated a matrix of the form:
$$
R = \left(\begin{array}{cc} R_2^2 & R_3^2 \\ R_2^3 & R_3^3 \end{array} \right) \, ,
$$
but because $R_{\gamma_a}(0) ( \mathbf{e}_2 ) = R(\dot{\gamma}_a(0),  \mathbf{e}_2) \dot{\gamma}_a(0) = R( \mathbf{e}_0 +  \mathbf{e}_1, \mathbf{e}_2)( \mathbf{e}_0 +  \mathbf{e}_1)$, we get:
$$
R_b^a = \mathbf{R}_{0b0}^a + \mathbf{R}_{1b0}^a + \mathbf{R}_{0b1}^a + \mathbf{R}_{1b1}^a \, .
$$ 

\item Compute $D_\alpha$ and $\zeta_\alpha$:
$$
D_\alpha = \det \left( R_\alpha - \frac{1}{2} \mathrm{Tr\,} (R_\alpha) \cdot \mathrm{Id} \right) = - \frac{1}{4}(R_2^2 - R_3^3)^2 - R_3^2R_2 ^3 \, ; \qquad \zeta_\alpha = \sqrt[4]{|D_\alpha |} \, .
$$

\item Compute $D_\alpha'$:
Compute the derivatives: $dR_b^a /dt$ using the covariant derivative $\nabla R$, that is:
$$
(R')_b^a := \frac{dR_a^b}{dt} = \frac{d\mathbf{R}_{0b0}^a}{dt} + \frac{d\mathbf{R}_{1b0}^a}{dt} + \frac{d\mathbf{R}_{0b1}^a}{dt} + \frac{d\mathbf{R}_{1b1}^a}{dt} \, , 
$$
and
$$
\frac{d\mathbf{R}_{abc}^d}{dt} = (\nabla R)_{abc0}^d + (\nabla R)_{abc1}^d \, .
$$
Then:
$$
D'_\alpha = - \frac{1}{2}\left( (R')_2^2 - (R')_3^3 \right) (R_2^2 - R_3^3) - (R')_3^2 R_2^3 - R_3^2 (R')_2^3  \, .
$$
The expressions for $D'_\alpha$ and $D''_\alpha$ can be given in compact form as follows.  Let $Q$ denote the matrix $Q = R_\alpha - \frac{1}{2} \mathrm{Tr\,}(R_\alpha) \cdot \mathrm{Id}$, then ($|Q| = \det Q$):
\begin{equation}\label{eq:D'}
D_\alpha' = |Q| \mathrm{Tr\,}(Q^{-1} Q') \, ,
\end{equation}
and
\begin{equation}\label{eq:D''}
D_\alpha'' = |Q| \left[ (\mathrm{Tr\,} (Q^{-1} Q'))^2 - \mathrm{Tr\,} ((Q^{-1}Q')^2) + \mathrm{Tr}(Q^{-1} Q'') \right] \, .
\end{equation}

\item Compute $D_\alpha''$:  Using Eq. (\ref{eq:D''}) we get:
\begin{eqnarray*}
D_\alpha'' &=& - \frac{1}{2} \left[ ((R')_2^2)^2 + R_2^2 (R'')_2^2 - 2(R')_2^2 (R')_3^3 - R_2^2(R'')_3^3 - R_3^3 (R'')_2^2 + ((R')_3^3)^2 + R_3^3 (R'')_3^3 \right] \\ & & - (R'')_3^2 R_2^3 - 2 (R')_3^2 (R')_2^3 - R_3^2 (R'')_2^3 \, .
\end{eqnarray*}

\item Using Eq. (\ref{eq:Cp}) we get:
$$
C =  \left( \begin{array}{cc} C_A^A & C_A^B \\ C_B^A & C_B^B \end{array}\right) = \frac{1}{\sqrt{D_\alpha}}  \left( \begin{array}{cc} R_2^2 & R_3^2 \\ R_2^3 & R_3^3 \end{array}\right) - \Phi I_2 \, ,
$$
with 
$$
\Phi = \frac{1}{8 \sqrt{|D_\alpha|}} \left[ \frac{9}{8} \left(\frac{D_\alpha'}{D_\alpha} \right)^2 - \frac{D_\alpha''}{D_\alpha}\right] \, .
$$
\end{enumerate}



\section{Schwarzschild's spacetime}\label{sec:schwarzschild}

We will illustrate the previous ideas considering Schwarzschild's spacetime (see, for instance, \cite{Ha73,Mo11}), that is, the four dimensional manifold $M$ equipped with the metric $g$ defined by:
\[
ds^2 = -\left(1-\frac{2m}{r}\right)dt^2 + \left(1-\frac{2m}{r}\right)^{-1}dr^2 + r^2 \left( d\phi^2 + \sin^2 \phi ~d\theta^2 \right)
\]
where $t\in \mathbb{R}$, $r\in \left(0,\infty\right)$, $\phi\in \left[0,2\pi\right)$, and $\theta\in \left[0,\pi\right)$.   In the indicated coordinates, the metric is singular at $r = 0$ and $r = 2m$.  We will call the connected manifold $r > 2m$, the exterior Schwarzschild spacetime, and the connected manifold $0 < r < 2m$, the interior Schwarzschild spacetime.

The non-vanishing components $\mathbf{R}^\sigma_{\mu\nu\rho}$  of the curvature tensor $R$, with $R\left(\frac{\partial}{\partial x^\nu},\frac{\partial}{\partial x^\rho}\right)\frac{\partial}{\partial x^\mu}= \mathbf{R}^\sigma_{\mu\nu\rho}\frac{\partial}{\partial x^\sigma}$, are
\[
\begin{array}{lll}
\mathbf{R}^t_{rtr} = \frac{2m}{r^2(r-2m)} & \mathbf{R}^t_{\phi t \phi} = -\frac{m}{r} & \mathbf{R}^t_{\theta t \theta} = -\frac{m \sin^2\phi}{r} \\
& & \\
\mathbf{R}^r_{trt} = -\frac{2m(r-2m)}{r^4} & \mathbf{R}^r_{\phi r \phi} = -\frac{m}{r} & \mathbf{R}^r_{\theta r \theta} = -\frac{m \sin^2\phi}{r} \\
& & \\
\mathbf{R}^{\phi}_{t \phi t} = \frac{m(r-2m)}{r^4} & \mathbf{R}^{\phi}_{r \phi r} = -\frac{m}{r^2(r-2m)} & \mathbf{R}^{\phi}_{\theta \phi \theta} = \frac{2m \sin^2\phi}{r} \\
& & \\
\mathbf{R}^{\theta}_{t \theta t} = \frac{m(r-2m)}{r^4} & \mathbf{R}^{\theta}_{r \theta r} = -\frac{m}{r^2(r-2m)} & \mathbf{R}^{\theta}_{\phi \theta \phi} = \frac{2m}{r} 
\end{array}
\]
where the obvious identification of the subindexes $\mu$ with the symbols denoting the coordinates $t,r,\theta, \phi$, has been used.


\subsection{The exterior Schwarzschild spacetime}\label{sec:null-Schwartzschild}

We will assume that $r>2m$ and then $\frac{r-2m}{r}>0$.  A lightlike vector $\alpha =\alpha^\mu \frac{\partial}{\partial x^\mu}$ in Schwarzschild spacetime, with $x^0=t$, $x^1=r$, $x^2=\phi$ and $x^3=\theta$, must verify 
\[
 -\left(\frac{r-2m}{r}\right)~(\alpha^0)^2 + \left(\frac{r}{r-2m}\right)~(\alpha^1)^2 + r^2 ~(\alpha^2)^2 + r^2\sin^2 \phi ~(\alpha^3)^2 =0
\] 
then the null direcctions are caracterized by $\alpha^1,\alpha^2,\alpha^3$ when $\alpha^0$ is fixed. 
Let us fix $\alpha^0=1$, like in Eq. (\ref{eq:alphaAB}), then a null direction is defined by:
\begin{equation}\label{eq-beta}
\left\{
\begin{array}{l}
\alpha^0 = 1  \\
\\
\alpha^1 =\frac{r-2m}{r} \cos B \sin A \\
\\
\alpha^2 = \frac{1}{r}\sqrt{\frac{r-2m}{r}} \sin B \sin A \\
\\
\alpha^3 = \frac{1}{r\sin \phi}\sqrt{\frac{r-2m}{r}} \cos A \, ,
\end{array}
\right.
\end{equation} 
with $A,B$ denoting polar angles, $0\leq A < \pi$, $0\leq B < 2\pi$.

Consider the orthonormal basis $\boldsymbol{\varepsilon}=\left( \varepsilon_0, \varepsilon_1, \varepsilon_2, \varepsilon_3 \right)$ given by:
\[
\varepsilon_0=\sqrt{\frac{r}{r-2m}}\frac{\partial}{\partial t}, \quad \varepsilon_1=\sqrt{\frac{r-2m}{r}}\frac{\partial}{\partial r}, \quad \varepsilon_2=\frac{1}{r}\frac{\partial}{\partial \phi}, \quad \varepsilon_3=\frac{1}{r\sin \phi}\frac{\partial}{\partial \theta}
\]
then, at any $p=(t,r,\phi,\theta)$ the map given by $(A,B)\mapsto \alpha(A,B)$:
\[
\alpha(A,B)= \varepsilon_0 + \cos B \sin A \cdot\varepsilon_1 + \sin B \sin A \cdot\varepsilon_2 +\cos A \cdot\varepsilon_3 \in \mathbb{PN}_p = S(p)
\]
is a parametrisation of the sky $S(p)$. 

If we denote by $\boldsymbol{\alpha}$ the column vector in $\mathbb{R}^3$ with components $\cos B \sin A$, $\sin B \sin A$, $\cos A$, Eq. (\ref{eq:alpha}), then we get:
\begin{equation}\label{eq:partialalpha}
\frac{\partial\boldsymbol{\alpha}}{\partial A}=\begin{pmatrix}
\cos B \cos A \\
\sin B \cos A \\
-\sin A
\end{pmatrix}, \quad 
\frac{\partial\boldsymbol{\alpha}}{\partial B}=\begin{pmatrix}
-\sin B \sin A \\
\cos B \sin A \\
0
\end{pmatrix}
\end{equation}
and we obtain (recall Eq. (\ref{eq:epsilon2})):
\[
\left( \mathbf{e}_0, \mathbf{e}_1, \mathbf{e}_2, \mathbf{e}_3 \right)= \boldsymbol{\varepsilon} \cdot \begin{pmatrix}
1 & 0 & 0 & 0 \\
\mathbf{0} & \boldsymbol{\alpha} & \frac{\partial\boldsymbol{\alpha}}{\partial A} & \frac{1}{\sin A}\frac{\partial\boldsymbol{\alpha}}{\partial B}
\end{pmatrix}
\]
is an orthonormal basis in $T_p M$.

Let us denote by $\gamma_{\alpha}$ the null geodesic such that $\gamma'_{\alpha}(0)=\alpha(A,B)=\mathbf{e}_0 + \mathbf{e}_1$ then, since the basis $ \left( \mathbf{e}_\mu \right)_{\mu=0,1,2,3}$ is orthonormal, we can consider $\left( \mathbf{e}_2, \mathbf{e}_3 \right)$ as a basis of $\langle \gamma_{\alpha}(0)^{\perp} \rangle$.

\subsubsection{The parametric curvature}

Consider $J = \mu \mathbf{e}_2 + \lambda \mathbf{e}_3 \in  \langle \gamma_{\alpha}(0)^{\perp} \rangle$ then, in order to calculate the parametric curvature, we will have to compute 
\[
R\left(\mathbf{e}_0 + \mathbf{e}_1,J\right)(\mathbf{e}_0 + \mathbf{e}_1) = \mu R\left(\mathbf{e}_0 + \mathbf{e}_1,\mathbf{e}_2\right)(\mathbf{e}_0 + \mathbf{e}_1)  + \lambda R\left(\mathbf{e}_0 + \mathbf{e}_1,\mathbf{e}_3\right)(\mathbf{e}_0 + \mathbf{e}_1)\in \{\gamma'(s)^{\perp}\}
\],

First, observe that for $i=2,3$
\begin{align*}
R\left(\mathbf{e}_0 + \mathbf{e}_1,\mathbf{e}_i\right)(\mathbf{e}_0 + \mathbf{e}_1) & = R\left(\mathbf{e}_0 ,\mathbf{e}_i\right)\mathbf{e}_0 + R\left(\mathbf{e}_0 ,\mathbf{e}_i\right)\mathbf{e}_1 + R\left(\mathbf{e}_1,\mathbf{e}_i\right)\mathbf{e}_0 + R\left(\mathbf{e}_1,\mathbf{e}_i\right)\mathbf{e}_1 = \\
& = \left(P^n_{00i} + P^n_{10i} +P^n_{01i} + P^n_{11i}  \right)\mathbf{e}_n
\end{align*}
where $P^n_{kij}$ denotes de components of the Riemann curvature related to the basis $(\mathbf{e}_i)_{i=1,\ldots , 4}$, that is:
\[
R\left(\mathbf{e}_i,\mathbf{e}_j\right)\mathbf{e}_k = P^n_{kij}\mathbf{e}_n
\]

So we have:
\begin{align*}
R\left(\mathbf{e}_0 + \mathbf{e}_1,\mathbf{e}_2\right)(\mathbf{e}_0 + \mathbf{e}_1)&= \frac{-3m}{r^3} \sin A ~\cos A ~\cos^2 B \cdot\left(\mathbf{e}_0 + \mathbf{e}_1 \right) + \\
& + \frac{3m}{r^3}\left( -1 + \sin^2 A ~\cos^2 B + 2\sin^2 B \right)\cdot\mathbf{e}_2 + \frac{6m}{r^3}\left( \cos A ~\cos B ~\sin B \right) \cdot\mathbf{e}_3 \\
\\
R\left(\mathbf{e}_0 + \mathbf{e}_1,\mathbf{e}_3\right)(\mathbf{e}_0 + \mathbf{e}_1)&= \frac{3m}{r^3} \sin A ~\cos B ~\sin B \cdot\left(\mathbf{e}_0 + \mathbf{e}_1 \right) + \\
& + \frac{6m}{r^3}\left(\cos A ~\cos B ~\sin B \right)\cdot\mathbf{e}_2 + \frac{3m}{r^3}\left(  1 - \sin^2 A ~\cos^2 B -2\sin^2 B  \right) \cdot\mathbf{e}_3
\end{align*}

Then, subtracting the component in the direction $\gamma'_{\alpha}(0)=\mathbf{e}_0 + \mathbf{e}_1$ we obtain 
\begin{align*}
\langle R\left(\mathbf{e}_0 + \mathbf{e}_1,\mathbf{e}_2\right)(\mathbf{e}_0 + \mathbf{e}_1) \rangle &= \frac{3m}{r^3}\left( -1 + \sin^2 A ~\cos^2 B + 2\sin^2 B\right)\cdot\mathbf{e}_2 + \frac{6m}{r^3}\left( \cos A ~\cos B ~\sin B \right) \cdot\mathbf{e}_3 \\
\\
\langle R\left(\mathbf{e}_0 + \mathbf{e}_1,\mathbf{e}_3\right)(\mathbf{e}_0 + \mathbf{e}_1) \rangle &=   \frac{6m}{r^3}\left(\cos A ~\cos B ~\sin B\right)\cdot\mathbf{e}_2 + \frac{3m}{r^3}\left(  1 - \sin^2 A ~\cos^2 B - 2\sin^2 B   \right) \cdot\mathbf{e}_3
\end{align*}
and therefore the parametric curvature map is
\[
\boxed{
R_{\alpha}(0)\left(\langle J \rangle \right)=\frac{3m}{r^3} \begin{pmatrix}
 -1 + \sin^2 A ~\cos^2 B + 2\sin^2 B  & 2 \cos A ~\cos B ~\sin B   \\
\\
 2\cos A ~\cos B ~\sin B  &  1- 2\sin^2 B  - \sin^2 A ~\cos^2 B  
\end{pmatrix} 
\begin{pmatrix}
\mu  \\
\lambda 
\end{pmatrix}
}
\]
where $\langle J \rangle \simeq 
\begin{pmatrix}
\mu  \\
\lambda 
\end{pmatrix}$, in the basis $\mathbf{e}_2,\mathbf{e}_3$, Eq. (\ref{eq:lambdamu}).

So, we get
\[
\boxed{ \rho_{\alpha}(0) = \mathrm{tr}\left(R_{\alpha}(0)\right) = 0  }
\]
and 
\[
\boxed{ D_{\alpha}(0) = \mathrm{det}\left(R_{\alpha}(0)\right) = -\left(\frac{3m (1- \sin^2 A ~\cos^2 B)}{r^3}\right)^2  }
\]

\subsubsection{The sky curvature}

Denote by $(\mathbf{E}_i)_{i=1,2,3,4}$ the parallel frame transported from $(\mathbf{e}_i)_{i=1,2,3,4}$ along $\gamma_{\alpha}$. 
Without any lack of generality, the computation of $R_{\alpha}(\tau)$ at $\gamma_{\alpha}(\tau)$ is still valid, so we have
\[
R_{\alpha}(\tau)=\frac{3m}{r^3} \begin{pmatrix}
 -1 + \sin^2 A ~\cos^2 B + 2\sin^2 B  & 2 \cos A ~\cos B ~\sin B   \\
\\
 2\cos A ~\cos B ~\sin B  &  1- 2\sin^2 B  - \sin^2 A ~\cos^2 B  
\end{pmatrix}
\]
where $A$, $B$ and $r$ are functions of the parameter $\tau$.

Now, we will compute the sky curvature $C_\Gamma (p)$ using Eq. (\ref{eq:Cp}).  We can write  
\[
R_{\alpha}(\tau)=\begin{pmatrix}
 R_1^1  & R_1^2   \\
R_2^1  &  R_2^2
\end{pmatrix} 
\]
where
\begin{align*}
R_1^1(\tau) & = \mathbf{g}\left( \langle R\left(\mathbf{E}_0 + \mathbf{E}_1,\mathbf{E}_2\right)(\mathbf{E}_0 + \mathbf{E}_1) \rangle ~, ~\mathbf{E}_2 \right) = P^{2}_{002}(\tau)+P^{2}_{102}(\tau)+P^{2}_{012}(\tau)+P^{2}_{112}(\tau)  \\
R_1^2(\tau) & = \mathbf{g}\left( \langle R\left(\mathbf{E}_0 + \mathbf{E}_1,\mathbf{E}_2\right)(\mathbf{E}_0 + \mathbf{E}_1) \rangle ~, ~\mathbf{E}_3 \right) = P^{3}_{002}(\tau)+P^{3}_{102}(\tau)+P^{3}_{012}(\tau)+P^{3}_{112}(\tau) \\
R_2^1 (\tau) & = \mathbf{g}\left( \langle R\left(\mathbf{E}_0 + \mathbf{E}_1,\mathbf{E}_3\right)(\mathbf{E}_0 + \mathbf{E}_1) \rangle ~, ~\mathbf{E}_2 \right) = P^{2}_{003}(\tau)+P^{2}_{103}(\tau)+P^{2}_{013}(\tau)+P^{2}_{113}(\tau) \\
R_2^2 (\tau) & = \mathbf{g}\left( \langle R\left(\mathbf{E}_0 + \mathbf{E}_1,\mathbf{E}_3\right)(\mathbf{E}_0 + \mathbf{E}_1) \rangle ~, ~\mathbf{E}_3 \right) = P^{3}_{003}(\tau)+P^{3}_{103}(\tau)+P^{3}_{013}(\tau)+P^{3}_{113}(\tau)
\end{align*}

Since $D_{\alpha}(\tau)=\det R_{\alpha} (\tau) $, then (recall Eqs. (\ref{eq:D'}), (\ref{eq:D''})):
\begin{align*}
D'_{\alpha}(\tau)&=\mathrm{tr}\left(\mathrm{Adj}(R_{\alpha}(\tau))R'_{\alpha}(\tau)\right)=\vert R_{\alpha} (\tau)\vert\cdot \mathrm{tr}\left(R^{-1}_{\alpha}(\tau)R'_{\alpha}(\tau)\right) \\
D''_{\alpha}(\tau)&=\vert R_{\alpha}(\tau) \vert\cdot \left[ \left(\mathrm{tr}\left(R^{-1}_{\alpha}(\tau)R'_{\alpha}(\tau)\right)\right)^2 - \mathrm{tr}\left(\left[R^{-1}_{\alpha}(\tau)R'_{\alpha}(\tau)\right]^2\right)+ \mathrm{tr}\left(R^{-1}_{\alpha}(\tau)R''_{\alpha}(\tau)\right) \right]
\end{align*}
where the prime $'$ means $\frac{d}{d\tau}$ . 
Hence, using Eq. (\ref{eq:CD}), we get:
\begin{equation}\label{eq-formula2-sky-curvat-R}
\boxed{ C_{\Gamma}\left(\gamma_{\alpha}\right)= \frac{1}{\sqrt{D_{\alpha}}}\left( R_{\alpha}+\frac{1}{8}\left[\mathrm{tr}\left(R^{-1}_{\alpha}R''_{\alpha}\right) - \mathrm{tr}\left(\left[R^{-1}_{\alpha}R'_{\alpha}\right]^2\right) -\frac{1}{8}\left( \mathrm{tr}\left(R^{-1}_{\alpha}R'_{\alpha}\right)\right)^2\right]\mathrm{Id}\right) }.
\end{equation}

In order to compute $R'_{\alpha}(0)$ and $R''_{\alpha}(0)$, we can observe that, for $j=2,3$:
\begin{align*}
 \left( \nabla_{\mathbf{E}_0 + \mathbf{E}_1}R\right) \left(\mathbf{E}_i,\mathbf{E}_j\right)\mathbf{E}_k & =  \nabla_{\mathbf{E}_0 + \mathbf{E}_1}\left(R\left(\mathbf{E}_i,\mathbf{E}_j\right)\mathbf{E}_k \right)- R\left(\nabla_{\mathbf{E}_0 + \mathbf{E}_1}\mathbf{E}_i,\mathbf{E}_j\right)\mathbf{E}_k - \\
 & - R\left(\mathbf{E}_i,\nabla_{\mathbf{E}_0 + \mathbf{E}_1}\mathbf{E}_j\right)\mathbf{E}_k - R\left(\mathbf{E}_i,\mathbf{E}_j\right)\nabla_{\mathbf{E}_0 + \mathbf{E}_1}\mathbf{E}_k 
\end{align*}
and, since $(\mathbf{E}_i)_{i=1,2,3,4}$ are parallel along $\gamma_{\alpha}$ and $\gamma'_{\alpha}(\tau)=\mathbf{E}_0 + \mathbf{E}_1$, then 
\begin{equation*}
 \left( \nabla_{\mathbf{E}_0 + \mathbf{E}_1}R\right) \left(\mathbf{E}_i,\mathbf{E}_j\right)\mathbf{E}_k  =  \nabla_{\mathbf{E}_0 + \mathbf{E}_1}\left(R\left(\mathbf{E}_i,\mathbf{E}_j\right)\mathbf{E}_k \right) = \frac{d}{d\tau}P^{n}_{kij}\cdot \mathbf{E}_{n}
\end{equation*}
and analogously we also get
\begin{equation*}
 \left( \nabla_{\mathbf{E}_0 + \mathbf{E}_1}\nabla_{\mathbf{E}_0 + \mathbf{E}_1}R\right) \left(\mathbf{E}_i,\mathbf{E}_j\right)\mathbf{E}_k  =   \frac{d^2}{d\tau^2}P^{n}_{kij}\cdot \mathbf{E}_{n} .
\end{equation*}
Therefore, at $\tau=0$ we have
\begin{equation*}
 \frac{d P^{n}_{kij}}{d\tau}(0)\cdot \mathbf{e}_{n} = \left( \nabla_{\mathbf{e}_0 + \mathbf{e}_1}R\right) \left(\mathbf{e}_i,\mathbf{e}_j\right)\mathbf{e}_k 
\end{equation*}
and analogously we also get:
\begin{equation*}
 \frac{d^2 P^{n}_{kij}}{d\tau^2}(0)\cdot \mathbf{e}_{n} =  \left( \nabla_{\mathbf{e}_0 + \mathbf{e}_1}\nabla_{\mathbf{e}_0 + \mathbf{e}_1}R\right) \left(\mathbf{e}_i,\mathbf{e}_j\right)\mathbf{e}_k   .
\end{equation*}

Since we only have to compute $P^{n}_{kij}$ and its derivatives for $n,j=2,3$ and $k,i=0,1$, then, the non--zero components are:
\begin{align*}
P^{2}_{002}(0)& = -P^{3}_{113}(0) =\frac{m}{r^3}\left(-2+3\sin^2 A + 3\sin^2 B - 3\sin^2 A~\sin^2 B \right) \, , \\
P^{2}_{112}(0)& = -P^{3}_{003}(0) =\frac{m}{r^3}\left(3\sin^2 B - 1\right)  \, ,  \\
P^{3}_{002}(0)& =  P^{3}_{112}(0) = P^{2}_{003}(0)=  P^{2}_{113}(0) = \frac{3m}{r^3}\cos A ~\cos B ~\sin B  \, ,
\end{align*}
and also:
\begin{align*}
\frac{d P^{2}_{002}}{d\tau}(0)& = \frac{3m}{r^{4}}\sqrt{\frac{r-2m}{r}}\cos B ~\sin A\left(5\cos^2 A~\cos^2 B -1 \right) \, , \\
\frac{d P^{2}_{112}}{d\tau}(0)& = \frac{3m}{r^{4}}\sqrt{\frac{r-2m}{r}}\cos B ~\sin A\left(5\cos^2 B -4 \right)\, , \\
\frac{d P^{3}_{002}}{d\tau}(0)& = \frac{15m}{r^{4}}\sqrt{\frac{r-2m}{r}}\cos^2 B ~\sin A\cos A ~\sin B  \, ,
\end{align*}
and, finally: 
\begin{align*}
\frac{d^2 P^{2}_{002}}{d\tau^2}(0)& = \frac{-3m}{r^6} \left[ r-3m +(13m-5r)\cos^2 B -2m \sin^2 A ~\cos^2 B - (75m-35r)\sin^2 A ~\cos^2 A ~\cos^4 B\right] \, ,\\
\frac{d^2 P^{2}_{112}}{d\tau^2}(0)& = \frac{-3m}{r^6} \left[ 3m - r -2m\sin^2 B -(5r-11m) \cos^2 A ~\cos^2 B + (75m-35r)\sin^2 A ~\sin^2 B ~\cos^2 B\right] \, , \\
\frac{d^2 P^{2}_{102}}{d\tau^2}(0)& =\frac{-3m^2}{r^6}\left(1-\sin^2 A ~\cos^2 B -2\sin^2 B\right) \, , \\
\frac{d^2 P^{3}_{002}}{d\tau^2}(0)& = \frac{-3m}{r^6} \left[ 62m -30r +(35r-75m)\left(1-\sin^2 A ~\cos^2 B\right)\right] \, , \\
\frac{d^2 P^{3}_{102}}{d\tau^2}(0)& = \frac{6m^2}{r^6}\cos A ~\cos B ~\sin B \, .
\end{align*}

Therefore, we obtain:
\begin{eqnarray*}
R'_{\alpha}(0) &=& \begin{pmatrix}
 (R')_1^1 (0)  & (R')_1^2 (0)  \\
(R')_2^1 (0)  &  (R')_2^2 (0)
\end{pmatrix}  = \\ &=& \frac{15m}{r^{4}}\sqrt{\frac{r-2m}{r}}\cos B ~\sin A \cdot \begin{pmatrix}
\cos^2 A ~\cos^2 B - \sin^2 B & 2\cos A ~\cos B ~\sin B     \\
2\cos A ~\cos B ~\sin B  &  \sin^2 B - \cos^2 A ~\cos^2 B 
\end{pmatrix} \, ,
\end{eqnarray*}
and
\begin{eqnarray*}
R''_{\alpha}(0) &=& \begin{pmatrix}
 (R'')_1^1 (0)  & (R'')_1^2 (0)   \\
(R'')_2^1 (0)  &  (R'')_2^2 (0)
\end{pmatrix} = \\ &=& \frac{15m}{r^{6}}\cdot f(r,A,B) \cdot \begin{pmatrix}
2\sin^2 B -1+ \sin^2 A ~\cos^2 B    & 2\cos A ~\cos B ~\sin B     \\
2\cos A ~\cos B ~\sin B  &  1- \sin^2 A ~\cos^2 B -2\sin^2 B
\end{pmatrix} \, ,
\end{eqnarray*}
where,
\[
f(r,A,B)= 6(r-2m) + (15m-7r)(1- \sin^2 A ~\cos^2 B) \, .
\]

Finally, if $p=p(A,B)=\gamma_{\alpha}(0)$, then substituting in (\ref{eq-formula2-sky-curvat-R}), we obtain:
\[
\boxed{ 
C_{\Gamma}(p) = \begin{pmatrix}
 \frac{1}{48m} \left[ (2m-15r) -\frac{5(r-2m)-96m\sin^2 B}{1-\cos^2 B ~\sin^2 A} \right] &  & \frac{2\cos A ~\sin B ~\cos B}{1-\cos^2 B ~\sin^2 A} \\
\\
 \frac{2\cos A ~\sin B ~\cos B}{1-\cos^2 B ~\sin^2 A} & & \frac{1}{48m} \left[ (98m-15r) -\frac{5(r-2m)+96m\sin^2 B}{1-\cos^2 B ~\sin^2 A} \right]
\end{pmatrix}
}
\]
and 
\begin{equation}\label{rhoext}
\boxed{
\rho_{\Gamma}(p) = \mathrm{tr}\left(C_{\Gamma}(p)\right) =  \frac{5}{24m} \left(10m-3r -\frac{r-2m}{1- \cos^2 B ~\sin^2 A }\right) 
}
\end{equation}
\[
\boxed{ \delta_{\Gamma} (p) =\mathrm{det}\left(C_{\Gamma}(p)\right)= \left[\frac{5}{48m}\left(10m-3r -\frac{r-2m}{1- \cos^2 A ~\sin^2 B }\right) \right]^2 - 1  }
\]

Therefore, the components of the $(1,1)$--tensor $C_{\mathrm{ext}}$ such that $(A,B)\mapsto \mathcal{C}_{\Gamma}(p)$, Eq. (\ref{eq:C}), in the exterior region of Schwarzschild spacetime, are given by:
\[
\left\{ 
\begin{tabular}{l}
$C^A_A  =  \frac{1}{48m} \left[ (2m-15r) -\frac{5(r-2m)-96m\sin^2 B}{1-\cos^2 B ~\sin^2 A} \right] = \left(\frac{\rho_{\Gamma}}{2}- 1\right) + \frac{2\sin^2 B}{1-\cos^2 B ~\sin^2 A}  $  \\
\\
$C^A_B = C_B^A =  \frac{2\cos A ~\sin B ~\cos B}{1-\cos^2 B ~\sin^2 A}  $  \\
\\
$C^B_B = \frac{1}{48m} \left[ (98m-15r) -\frac{5(r-2m)+96m\sin^2 B}{1-\cos^2 B ~\sin^2 A} \right] = \left( \frac{\rho_{\Gamma}}{2}+ 1\right)- \frac{2\sin^2 A}{1-\cos^2 B ~\sin^2 A}$  \\
\end{tabular}
\right.
\]
In this situation we get the conformal sign $\epsilon = -1$.  In fact from Eqs. (\ref{rhoext}), (\ref{eq:epsilon}), we get $\epsilon = \delta - \rho^2/4 = -1$. 


\subsection{The interior Schwarzschild spacetime}\label{sec:null-Schwartzschild-inner}

Now, we will assume that $0<r<2m$ and then $\frac{2m-r}{r}>0$.

The metric is defined by 
\[
ds^2 = - \left(\frac{2m}{r}-1\right)^{-1}dr^2 + \left(\frac{2m}{r}-1\right)dt^2 + r^2 \left( d\phi^2 + \sin^2 \phi ~d\theta^2 \right)
\]
and therefore $\frac{\partial}{\partial r}$ is timelike and $\frac{\partial}{\partial t}$ is spacelike.

We repeat the same calculation for $x^1=r$, $x^2=t$, $x^3=\phi$ and $x^4=\theta$ (Notice that we have swapped the roles between the variables $r$ and $t$).  
A lightlike vector $\alpha =\alpha^\mu \frac{\partial}{\partial x^\mu}$, must verify:
\[
 -\left(\frac{r}{2m-r}\right)~(\alpha^0)^2 + \left(\frac{2m-r}{r}\right)~(\alpha^1)^2 + r^2 ~(\alpha^2)^2 + r^2\sin^2 \phi ~(\alpha^3)^2 =0
\] 
then the null directions are characterized by $\alpha^1,\alpha^2,\alpha^3$ when $\beta^1$ is fixed. 
So, if we fix $\alpha^0=1$, then a null direction can be defined by:
\begin{equation}\label{eq-beta2}
\left\{
\begin{tabular}{l}
$\alpha^0 = 1 $ \\
\\
$\alpha^1 =\frac{r}{2m-r} \cos B \sin A$ \\
\\
$\alpha^2 = \frac{1}{r}\sqrt{\frac{r}{2m-r}} \sin B \sin A$ \\
\\
$\alpha^3 = \frac{1}{r\sin \phi}\sqrt{\frac{r}{2m-r}} \cos A $ .
\end{tabular}
\right.
\end{equation} 

An orthonormal basis $\boldsymbol{\varepsilon} = \left( \varepsilon_0, \varepsilon_1, \varepsilon_2, \varepsilon_3 \right)$ is given by
\[
\varepsilon_0=\sqrt{\frac{2m-r}{r}}\frac{\partial}{\partial r}, \quad \varepsilon_1=\sqrt{\frac{r}{2m-r}}\frac{\partial}{\partial t}, \quad \varepsilon_2=\frac{1}{r}\frac{\partial}{\partial \phi}, \quad \varepsilon_3=\frac{1}{r\sin \phi}\frac{\partial}{\partial \theta}
\]
then, at any $p=(r,t,\phi,\theta)$ the map given by $(A,B)\mapsto \alpha(A,B)$ where 
\[
\alpha(A,B)= \varepsilon_0 + \cos B \sin A \cdot\varepsilon_1 + \sin B \sin A \cdot\varepsilon_2 +\cos A \cdot\varepsilon_3 \in \mathbb{PN}_p = S(p)
\]
is a parametriwation of the sky $S(p)$.  Then we get that:
\[
\left( \mathbf{e}_0, \mathbf{e}_1, \mathbf{e}_2, \mathbf{e}_3 \right)= \boldsymbol{\varepsilon}\cdot \begin{pmatrix}
1 & 0 & 0 & 0 \\
\mathbf{0} & \boldsymbol{\alpha} & \frac{\partial\boldsymbol{\alpha}}{\partial A} & \frac{1}{\sin A}\frac{\partial\boldsymbol{\alpha}}{\partial B}
\end{pmatrix}
\]
is a orthonormal basis in $T_p M$, with $\boldsymbol{\alpha}$ as in Eq. (\ref{eq:alpha}).

Let us call $\gamma_{\alpha}$ the null geodesic such that $\gamma'_{\alpha}(0)=\alpha(A,B)=\mathbf{e}_0 + \mathbf{e}_1$ then, since the basis $ \left( \mathbf{e}_i \right)_{i=1,2,3,4}$ is orthonormal, we can consider $\left( \mathbf{e}_2, \mathbf{e}_3 \right)$ as a basis of $\langle \gamma_{\alpha}(0)^{\perp} \rangle$.

\subsubsection{The parametric curvature}

If we consider $J = \mu \mathbf{e}_2 + \lambda \mathbf{e}_3 \in  \langle \gamma_{\alpha}(0)^{\perp} \rangle$ then, using the same procedure as for the exterior Schwartzschild spacetime, we obtain again that the parametric curvature is
\[
\boxed{
R_{\alpha}(0)\left(\langle J \rangle \right)=\frac{3m}{r^3} \begin{pmatrix}
 -1 + \sin^2 A ~\cos^2 B + 2\sin^2 B  & 2 \cos A ~\cos B ~\sin B   \\
\\
 2\cos A ~\cos B ~\sin B  &  1- 2\sin^2 B  - \sin^2 A ~\cos^2 B  
\end{pmatrix} 
\begin{pmatrix}
\mu  \\
\lambda 
\end{pmatrix}
}
\]
where $\langle J \rangle \simeq 
\begin{pmatrix}
\mu  \\
\lambda 
\end{pmatrix}$.

Therefore
\[
\boxed{ \rho_{\alpha}(0) = \mathrm{tr}\left(R_{\alpha}(0)\right) = 0  }
\]
and 
\[
\boxed{ D_{\alpha}(0) = \mathrm{det}\left(R_{\alpha}(0)\right) = -\left(\frac{3m (1- \sin^2 A ~\cos^2 B)}{r^3}\right)^2  }
\]

\subsubsection{The sky curvature}

Again, in an analogous calculation we have that if $p=p(A,B)=\gamma_{\alpha}(0)$, then 
\[
\boxed{ 
C_{\Gamma}(p) = \begin{pmatrix}
 \frac{1}{48m} \left[ (12m-20r) -\frac{5(2m-r)-96m\sin^2 B}{1-\cos^2 B ~\sin^2 A} \right] &  & \frac{2\cos A ~\sin B ~\cos B}{1-\cos^2 B ~\sin^2 A} \\
\\
 \frac{2\cos A ~\sin B ~\cos B}{1-\cos^2 B ~\sin^2 A} & & \frac{1}{48m} \left[ (108m-20r) -\frac{5(2m-r)+96m\sin^2 B}{1-\cos^2 B ~\sin^2 A} \right]
\end{pmatrix}
}
\]
and 
\begin{equation}\label{rhoint}
\boxed{
\rho_{\Gamma}(p) = \mathrm{tr}\left(C_{\Gamma}(p)\right) =  \frac{5}{24m} \left(4\left(3m-r\right) - \frac{2m-r}{1- \cos^2 B ~\sin^2 A }\right) 
}
\end{equation}
\[
\boxed{ \delta_{\Gamma} (p) =\mathrm{det}\left(C_{\Gamma}(p)\right)= \left[\frac{5}{48m}\cdot\left(4\left(3m-r\right) - \frac{2m-r}{1- \cos^2 B ~\sin^2 A }\right) \right]^2 - 1  }
\]

Therefore, the components of the $(1,1)$--tensor $C_{\mathrm{int}}$ such that $(A,B)\mapsto C_{\Gamma}(p)$ in the interior region of Schwarzschild spacetime, are given by: 
\[
\left\{ 
\begin{tabular}{l}
$C^A_A  =  \frac{1}{48m} \left[ (12m-20r) -\frac{5(2m-r)-96m\sin^2 B}{1-\cos^2 B ~\sin^2 A} \right] = \left(\frac{\rho_{\Gamma}}{2}- 1\right) + \frac{2\sin^2 B}{1-\cos^2 B ~\sin^2 A}  $  \\
\\
$C^A_B = C_B^A = \frac{2\cos A ~\sin B ~\cos B}{1-\cos^2 B ~\sin^2 A}  $  \\
\\
$C^B_B = \frac{1}{48m} \left[ (108m-20r) -\frac{5(2m-r)+96m\sin^2 B}{1-\cos^2 B ~\sin^2 A} \right] = \left( \frac{\rho_{\Gamma}}{2}+ 1\right)- \frac{2\sin^2 A}{1-\cos^2 B ~\sin^2 A}$  \\
\end{tabular}
\right.
\]

\subsection{Continuity of $\rho$ and $\delta$ at the Schwarztschild radius $r=2m$}

Let us call 
\[
\chi=\chi(A,B)=1-\cos^2 B~\sin^2 A
\]
for both regions (interior and exterior) of Schwartzschild spacetime, although it does not have the same meaning because the vectors $\mathbf{E}_2$ and $\mathbf{E}_3$ can not be extended to both regions in a continuous way. And let us denote by $\rho_{\chi}(r)=\rho(r,\chi)$ the expressions (\ref{rhoext}) and (\ref{rhoint}).

But since the lateral limits of $\rho_{\chi}(p)$ at $r=2m$ exist and coincide
\[
\lim_{r\mapsto 2m^{-}} \rho_{\chi}(r)=\frac{5}{6} = \lim_{r\mapsto 2m^{+}} \rho_{\chi}(r)
\]
then $\rho_{\chi}(r)$ can be extended continuously to $r=2m$.

And because $\delta_{\chi}(r)=\left( \frac{\rho_{\chi}(r)}{2} \right)^2 -1$, then 
\[
\lim_{r\mapsto 2m^{-}} \delta_{\chi}(r)=-\frac{119}{144} = \lim_{r\mapsto 2m^{+}} \delta_{\chi}(r)
\]
so $\rho_{\chi}(r)$ also extends continuously to $r=2m$.

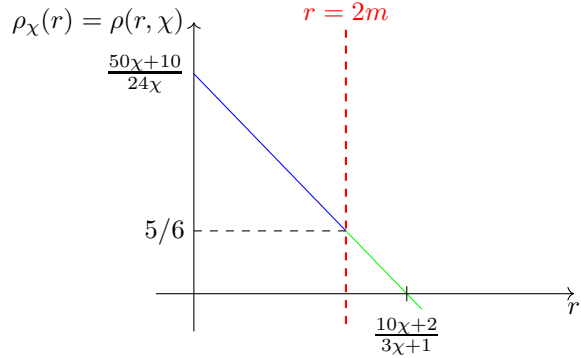
\begin{figure}[!h]
\begin{tikzpicture}[scale=1]
\draw[blue] (0,2.9166) -- (2,0.833) ;
\draw (0,2.9166) node[anchor=east] {$\frac{50\chi + 10}{24\chi}$}; 
\draw[green] (2,0.833) -- (3,-0.20833) ;
\draw (2.8,-0.2) node[anchor=north] {$\frac{10\chi + 2}{3\chi +1}$};
\draw (2.8,-0.1) -- (2.8,0.1);
\draw[->] (0,-0.5) -- (0,3.6) node[anchor=east] {$\rho_{\chi}(r)=\rho(r,\chi)$};
\draw[->] (-0.5,0) -- (5,0) node[anchor=north] {$r$};
\draw[red,thick,dashed] (2,3.5) node[anchor=south] {$r=2m$} -- (2,-0.5);
\draw[dashed] (0,0.833) node[anchor=east] {$5/6$} -- (2,0.833);
\end{tikzpicture}
\caption{Observe that $\lim_{r\mapsto 2m}\rho_{\chi}(r)=\frac{5}{6}$ for all $\chi=1-\cos^2 B~\sin^2 A \in \left(0,1\right]$ for both the \textcolor{blue}{interior} and the \textcolor{green}{exterior}.}
\end{figure}

\begin{figure}[!h]
\begin{tikzpicture}[scale=1]
\foreach \h in {1,0.3,0.02}{
      \draw ({((2+10*\h)/(1+3*\h))+(12/5)*(\h/(1+3*\h))},-0.5)  -- ({((2+10*\h)/(1+3*\h))-(84/5)*(\h/(1+3*\h))},3.5) ; 
      } 
\filldraw[color=white] (-2,2) rectangle (0,3.6);
\draw (3.6,-0.5) node[anchor=north west] {\tiny{$\chi=1$}};
\draw (3.01,-0.5) node[anchor=north] {\tiny{$\chi=0.3$}};
\draw (2,-0.5) node[anchor=north] {\tiny{$\chi=0.02$}};
\draw[->] (0,-0.5) -- (0,3.6) node[anchor=east] {$\rho_{\chi}(r)=\rho(r,\chi)$};
\draw[->] (-0.5,0) -- (5,0) node[anchor=north] {$r$};
\draw[red,thick,dashed] (2,3.5) node[anchor=south] {$r=2m$} -- (2,-0.5);
\draw[dashed] (0,0.833) node[anchor=east] {$5/6$} -- (2,0.833);
\draw[->,blue, bend left=45, dashed] (0.3,2.3) to (2,3.2); 
\draw[blue] (0.8,3.3) node {\tiny{$\chi \mapsto 0$}}; 
\draw (-0.1,2.5) node[anchor=east] {$5/2$} -- (0.1,2.5);
\draw (3,-0.1) -- (3,0.1) node[anchor=south] {$3$};
\end{tikzpicture}
\caption{The slope of the the straight line $y=\rho_{\chi}(r)$ decreases to $-\infty$ when $\chi \mapsto 0$.}
\end{figure}
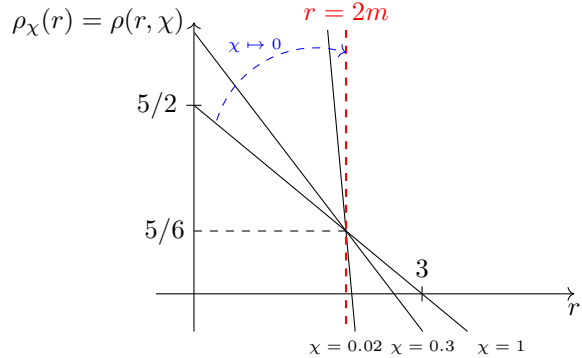

Notice that the lateral limits when $r\mapsto 2m$ of the sky matrices corresponding to the exterior and interior regions also coincide with 
\[
\boxed{ 
\lim_{r\mapsto 2m^{-}} C_{\Gamma}(p) = \begin{pmatrix}
 \frac{1}{48} \left[ -28 +\frac{96\sin^2 B}{1-\cos^2 B ~\sin^2 A} \right] &  & \frac{2\cos A ~\sin B ~\cos B}{1-\cos^2 B ~\sin^2 A} \\
\\
 \frac{2\cos A ~\sin B ~\cos B}{1-\cos^2 B ~\sin^2 A} & & \frac{1}{48} \left[ 68 -\frac{96\sin^2 B}{1-\cos^2 B ~\sin^2 A} \right]
\end{pmatrix} = \lim_{r\mapsto 2m^{+}} C_{\Gamma}(p) \, .
}
\]

\section{Conclusions and discussion}

A novel notion of conformal invariants has been introduced using the space of light rays $\mathcal{N}$ associated to a given spacetime $M$ to construct them.  More specifically, the canonical bundle over the spacetime $M$ whose fibre at each point is the corresponding sky has been used and the conformal invariants have been constructed exploiting the conformal geometry of light rays.   Each light ray $\Gamma$ carries a canonical conformal covariant derivative that can be used to construct an endomorphism $R_\gamma$ on the tangent spaces of the sky containing it.      In addition the existence of a distinguished parametrisation by a conformal invariant parameter $s$ is used to associate a family of scalar absolute conformal invariant or, equivalently, of sky conformal curvatures, to the given spacetime.   

The definition of the conformal parametrisation together with the transformation properties of the tensors $R_\gamma$ under reparametrisations allows us to construct an algorithm that can be implemented on any symbolic manipulation language and that has been successfully used to compute the sky curvatures of Schwarzschild spacetime.   

Other conformal invariants can be obtained from the sky curvature tensor, for instance, its principal directions.   In what sense these new conformal invariants characterise the conformal class of the original spacetime?

Another set of relevant questions emerge from the notion of conformal invariants themselves, as the definition provided in the paper suggest its extension to larger classes of objects, suitably described using a categorical language, subject that will be discussed elsewhere.


\section*{Acknowledgments}
The authors acknowledge financial support from the Spanish Ministry of Economy and Competitiveness, through the Severo Ochoa Programme for Centres of Excellence in RD (SEV-2015/0554), the MINECO research project  PID2020-117477GB-I00,  and Comunidad de Madrid project QUITEMAD++, S2018/TCS-A4342.


\newpage


\end{document}